\documentclass[11pt,hyp,]{nyjm}
\usepackage{hyperref}
\hypersetup{nesting=true,debug=true,naturalnames=true}
\usepackage{graphicx,amssymb,upref}

\let\<\langle
\let\>\rangle

\let\uml\"

\title{Bures Contractive Channels on Operator Algebras} 

\author{Douglas Farenick}  
\address{Department of Mathematics \& Statistics, University of Regina, Regina, Saskatchewan S4S 0A2, Canada} 
\email{douglas.farenick@uregina.ca}  
\thanks{Supported in part by an NSERC Discovery Grant} 

\author{Mizanur Rahaman}  
\address{Department of Mathematics \& Statistics, University of Regina, Regina, Saskatchewan S4S 0A2, Canada} 
\email{mizanur1@gmail.com}  
\thanks{Supported in part by a University of Regina Graduate Research Fellowship} 



\newtheorem{definition}{Definition}[section]
\newtheorem{theorem}[definition]{Theorem}
\newtheorem{lemma}[definition]{Lemma}
\newtheorem{corollary}[definition]{Corollary}
\newtheorem{proposition}[definition]{Proposition}
\theoremstyle{definition}

\newtheorem{defn}[definition]{Definition}
\newtheorem{example}[definition]{Example}

\newcommand\style{\mathsf }

\newcommand{\B}{\style{B}}
\newcommand{\M}{\style{M}}
\renewcommand{\L}{\mathcal{H}_{\rm aux}}
\newcommand\A{{\style A}}
\renewcommand{\H}{\style{H}}

\newcommand{\N}{\style{N}}

\newcommand{\Ep}{\mathcal{E}}

\newcommand{\lgG}{\rm{G}}        
\newcommand{\lgGL}{{\rm{GL}}}        

\newcommand\tr{ \operatorname{Tr} } 

\begin{document} 
 
\begin{abstract}  
In a unital C$^*$-algebra with a faithful trace functional $\tau$, the set $\mathcal D_\tau(\A)$ of positive $\rho\in\A$ of trace
$\tau(\rho)=1$ is an algebraic analogue of the space of density matrices (the set of all positive matrices of a fixed dimension
of unit trace). Motivated by the literature concerning the metric properties of the space of density matrices, the
present paper studies the density space $\mathcal D_\tau(\A)$ in terms of the Bures metric. Linear maps on $\A$
that map $\mathcal D_\tau(\A)$ back into itself are positive and trace preserving; hence, they may be viewed as
an algebraic analogue of a quantum channel, which are studied intensely in the literature on quantum computing
and quantum information theory. 

The main results in this paper are: (i) to establish that the Bures metric is indeed a metric; (ii) to prove that channels
induce nonexpansive maps of the density space $\mathcal D_\tau(\A)$; (iii) to introduce and study channels on $\A$
that are locally contractive maps (which we call \emph{Bures contractions})
on the metric space $\mathcal D_\tau(\A)$; and (iv) to analyse Bures contractions from the point of view of the Frobenius theory
of cone preserving linear maps.

Although the focus is on unital C$^*$-algebras, an important class of examples is furnished by finite von Neumann algebras. 
Indeed, several of the C$^*$-algebra results are established by first proving them for finite von Neumann algebras and then
proving them for C$^*$-algebras by embedding a C$^*$-algebra $\A$ into its enveloping von Neumann algebra $\A^{**}$.
\end{abstract} 
\maketitle
\tableofcontents

\section{Introduction}

\subsection{Motivation}

Density matrices, which are positive semidefinite matrices $\rho$ of a fixed dimension of trace $\tr(\rho)=1$, are an essential feature
of mathematical physics and quantum theory. The monograph of Bengtsson and {\.Z}yczkowski \cite{Bengtsson--Zyczkowski-book},
for example, is devoted to the geometry of spaces of density matrices and to applications to quantum entanglement. Certain applications
require the use of density operators (positive trace-class operators of unit trace) acting on an infinite-dimensional separable Hilbert space,
but the theory is slightly different in this context because the identity operator, in contrast to the matrix case, has infinite trace.
Another direction where commonalities with the matrix setting might be expected is with unital C$^*$-algebras that possess a
faithful trace functional in which the identity of the algebra has finite trace. This occurs, for example, 
with a large class of interesting separable C$^*$-algebras, and with finite von Neumann algebras. 
In this paper we consider both types of operator algebras. Specifically, if $\A$ is an operator algebra (that is, a C$^*$- or von Neumann algebra)
possessing a faithful (finite) trace functional $\tau$, then the set $\mathcal D_\tau(\A)$ of all positive elements $\rho\in\A$
with trace $\tau(\rho)=1$ is a norm-closed convex set called the \emph{$\tau$-density space} of $\A$.

There are several metrics of interest on spaces of density matrices. Our interest in the present paper is with the \emph{Bures metric},
which originates in a paper of Bures \cite{bures1969} and was later adapted to finite-dimensional Hilbert space by
Uhlmann \cite{uhlmann1976} and Jozsa \cite{josza1994}. Expository works on the Bures metric for density matrices 
are given in the monographs \cite{Bengtsson--Zyczkowski-book,Hayashi-book}, for example.

\begin{defn} If $\tau$ is a faithful tracial linear functional on a unital C$^*$-algebra $\A$, and if
$\mathcal D_\tau(\A)=\{\rho\in\A\,:\,\rho\mbox{ is positive and }\tau(\rho)=1\}$, then 
the \emph{Bures distance} $d_B^\tau(\sigma,\rho)$ between two elements $\sigma,\rho\in \mathcal D_\tau(\A)$ is defined to be
\[
d_B^\tau(\sigma,\rho)= \sqrt{1-\tau(|\sigma^{1/2}\rho^{1/2}|)}.  
\]
The quantity $\tau(|\sigma^{1/2}\rho^{1/2}|)$ is called the \emph{fidelity} of $\sigma$ and $\rho$  and
is denoted by $F_\tau(\sigma,\rho)$. 
\end{defn}

If $\tau$ is the canonical trace on the algebra $\M_d(\mathbb C)$ of
$d\times d$ complex matrices, then the Bures distance is well known to be
a metric on the space of $d\times d$ density matrices.
However, for the case of an arbitrary operator algebra $\A$, additional work must be carried out to show that
$d_B^\tau$ is a metric on  $\mathcal D_\tau(\A)$ (see Section \S\ref{bm}).

If $\Ep:\A\rightarrow\A$ is a linear transformation that maps $\mathcal D_\tau(\A)$ back into itself, then $\Ep$ is
necessarily positive and trace preserving; we will call such maps \emph{channels},
again motivated by the literature in mathematical physics. One point of departure from the mainstream literature, however,
is that we do not require our channels to be completely positive, which is very often the basic assumption in quantum 
information theory. In this we regard we are guided by the monograph of St{\o}rmer \cite{Stormer-book}, which demonstrates that
complete positivity is often an unnecessarily strong assumption and that one need only assume that the positive linear
maps in question satisfy the (generalised) Schwarz inequality. Indeed, our results in this paper are most often phrased for
maps of this Schwarz type.

Every channel $\Ep:\A\rightarrow\A$ induces a continuous affine function $f_\Ep$ on the metric space
$\left( \mathcal D_\tau(\A), d_B^\tau\right)$
defined by $f_\Ep(\rho)=\Ep(\rho)$. It so happens that the affine functions $f_\Ep$ are nonexpansive maps -- that is, 
$f_\Ep$ does not increase the Bures distance between pairs of density elements. 
In an earlier paper \cite{farenick--jaques--rahaman2016}, we studied the structure of those channels $\Ep:\A\rightarrow\A$ 
that induce isometric maps of the metric space $\left(\mathcal D_\tau(\A),d_B^\tau \right)$.
In the present paper we consider the opposite type of behaviour in that we study channels $\Ep$, called \emph{Bures contractions} herein,
for which $d_B^\tau(\Ep(\sigma),\Ep(\rho))<d_B^\tau(\sigma,\rho)$ for all pairs of distinct $\sigma,\rho\in \mathcal D_\tau(\A)$.  
Our interest is with the spectral properties of Bures contractive channels (from the point of view of Perron-Frobenius theory), and certain
mapping features of these channels, including multiplicative domains, fixed points, invariant faces of the positive cone $\A_+$ of $\A$, and inverses.  

In a much earlier paper \cite{raginsky2002}, M.~Raginsky studied contractive channels with respect to the trace-norm metric $d_1^\tau(\sigma,\rho)=\tau(|\sigma-\rho|)$
on the density-matrix space.
Although a number of our results are inspired by his work, there are nevertheless some notable differences: (i) we consider what are called locally contractive
linearly induced affine maps of the metric space $\left(\mathcal D_\tau(\A),d_B^\tau \right)$, whereas Raginsky is concerned with strictly contractive 
linearly induced affine maps of
$\left(\mathcal D_\tau(\A),d_1^\tau \right)$; and (ii) we consider general tracial operator algebras, whereas Raginsky is motivated by issues in quantum
computation and therefore is concerned only with matrices. The metrics $d_B^\tau$ and $d_1^\tau$ induce the same topologies on $\mathcal D_\tau(\A)$ (Proposition \ref{top equiv}),
but the analysis of contractions differs in light of the inequality $\sqrt{2}d_B^\tau(\sigma,\rho)\leq d_1^\tau(\sigma,\rho)$, for all $\sigma,\rho\in \mathcal D_\tau(\A)$
(Proposition \ref{fuchs}). We refer the reader to \cite{jian--he--yuan--wang2013} for a study of isometric channels with respect to the trace norm in the setting of finite-dimensional Hilbert space. 

\subsection{Definitions, assumptions, and notation}
The following definitions are standard, but are stated here explicitly for clarity.

\begin{definition}
A linear map $\Phi:\A\rightarrow\A$ of a unital C$^*$-algebra $\A$ is:
\begin{enumerate}
\item \emph{unital}, if $\Phi(1)=1$ (the multiplicative identity of $\A$);
\item \emph{positive}, if $\Phi(\A_+)\subseteq\A_+$;
\item \emph{$k$-positive}, if $\Phi^{(k)}:\M_k(\A)\rightarrow\M_k(\A)$ is a positive linear map, 
where $\Phi^{(k)}\left([a_{ij}]_{i,j}\right)=[\Phi(a_{ij})]_{i,j}$ for all 
$k\times k$ matrices $[a_{ij}]_{i,j}$ with entries $a_{ij}\in\A$;
\item \emph{completely positive}, if $\Phi$ is $k$-positive for every $k\in\mathbb N$; 
\item a \emph{Schwarz map}, if $\Phi(x)^*\Phi(x)\leq \Phi(x^*x)$, for every $x\in\A$.
\end{enumerate}
\end{definition}

Note that Schwarz maps are necessarily positive. While it is true that every
unital $2$-positive linear map is a Schwarz map, there do exist unital Schwarz maps that fail to be 
$2$-positive \cite{choi1980b}.

\begin{definition}[Channels in the C$^*$-algebra Category]
A linear map $\mathcal E:\A\rightarrow\A$ of a unital C$^*$-algebra $\A$ with a distinguished faithful trace functional $\tau$ is:
\begin{enumerate}
\item \emph{trace-preserving}, if $\tau\circ\mathcal E=\tau$;
\item a \emph{channel}, if $\mathcal E$ is trace-preserving and positive;
\item a \emph{$k$-positive channel}, if $\mathcal E$ is trace-preserving and $k$-positive;
\item a \emph{completely positive channel}, if $\mathcal E$ is trace-preserving and completely positive;
\item a \emph{Schwarz channel}, if $\mathcal E$ is a trace-preserving Schwarz map.
\end{enumerate}
\end{definition}

\begin{definition}[Channels in the Finite von Neumann algebra Category]
A positive linear map $\mathcal E:\N\rightarrow\N$ of a (finite) von Neumann algebra $\N$ with a distinguished normal faithful trace functional $\tau$ is:
\begin{enumerate}
\item \emph{normal}, if $\Ep\left(\sup_\lambda h_\lambda\right)=\sup_\lambda\Ep(h_\lambda)$, for every bounded-above increasing net
$\{h_\lambda\}_\lambda$ of selfadjoint operators $h_\lambda\in\N$;
\item a \emph{channel}, if $\mathcal E$ is trace-preserving and normal.
\end{enumerate}
\end{definition}

One natural occurrence of normal traces and channels
is as follows. If $\tau$ is a faithful trace functional on a unital C$^*$-algebra $\A$, then the bidual $\tau^{**}$ of $\tau$ is
a faithful normal trace on the enveloping von Neumann algebra (i.e., on the bidual) $\A^{**}$ of $\A$. Moreover, if $\Ep$ is
a channel on $(\A,\tau)$, then $\Ep^{**}$ is a channel on $(\A^{**},\tau^{**})$; indeed, $\Ep^{**}$
is a Schwarz channel if $\Ep$ is a Schwarz channel \cite[Lemma 3]{gardner1979a} .

To avoid repetitious statements of hypotheses, the following notational assumptions 
are made for the remainder
of the paper:

\begin{itemize}
\item $\A$ denotes a unital C$^*$-algebra, and $\A_+$ denotes the cone of positive elements of $\A$;
\item $\N$ denotes a finite von Neumann algebra;
\item $\tau$ denotes a fixed faithful trace on a unital C$^*$-algebra $\A$ or a fixed faithful normal trace on a finite von Neumann algebra $\N$;
\item $\lgGL(\A)_+$ denotes $\lgGL(\A)\cap\A_+$, the set of positive invertible elements.
\end{itemize}

 We record the following basic facts about positive linear maps for future reference.

\begin{proposition}\label{basic positive facts} Assume that $\mathcal E:\A\rightarrow\A$ is a linear transformation.
\begin{enumerate}
\item {\rm (Russo-Dye)} If $\Ep$ is positive, then $\Ep$ is continuous and $\|\Ep\|=\|\Ep(1)\|$.
\item $\Ep$ maps $\mathcal D_\tau(\A) $ back into itself if and only if $\Ep$ is a channel.
\item If $\Ep$ is a Schwarz channel, then $\Ep(1)=1$.
\end{enumerate}
\end{proposition}

\begin{proof} The Russo-Dye Theorem is proved in  \cite[Chapter 2]{Paulsen-book}, while the fact that 
$\Ep$ maps $\mathcal D_\tau(\A) $ back into itself if and only if
$\Ep$ is positive and $\tau\circ\Ep=\tau$ is easy to verify. 

To prove the third statement, note that for every $x\in\A$ we have that $\Ep(x^*x)\geq\Ep(x)^*\Ep(x)$, which implies that
\[
\|\Ep\|\,\|x\|^2 = \|\Ep\|\,\|x^*x\| \geq \|\Ep(x^*x)\|\geq \|\Ep(x)^*\Ep(x)\| =\|\Ep(x)\|^2.
\]
With $x=1$ and using $\|\Ep\|=\|\Ep(1)\|$, the inequality above yields $\|\Ep\|^2\leq\|\Ep\|$, and so $\|\Ep\|\leq 1$.
Thus, ${\mathcal E}(1)$ is a positive contraction and, therefore, $1-{\mathcal E}(1)$ is positive. Hence,
$0\leq\tau(1-{\mathcal E}(1))=\tau(1)-\tau\circ{\mathcal E}(1)=\tau(1)-\tau(1)=0$. Because $\tau$ is faithful and since
$1-{\mathcal E}(1)\in\A_+$, we
see that $\tau(1-{\mathcal E}(1))=0$ only if $1-{\mathcal E}(1)=0$. Therefore, $\Ep(1)=1$.
\end{proof}

\section{The Bures Metric}\label{bm}

This section establishes, for tracial C$^*$-algebras, a number of 
important properties of the trace functional and the Bure metric that are known to hold in matrix algebras.
Once again, the notation assumptions herein are that $\A$ denotes a unital C$^*$-algebra and $\tau$ denotes a faithful
trace on $\A$.

\subsection{The trace norm is a norm}

It is well known that the function $x\mapsto\tau(|x|)$ defines a norm on a (semi)finite von Neumann algebra
with faithful normal trace $\tau$.
This function also defines a norm on a unital C$^*$-algebra $\A$ with faithful trace $\tau$, as noted in the following
result.

\begin{proposition}\label{tn}
The function $x\mapsto\tau(|x|)$ defines a norm on $\A$.
\end{proposition}

\begin{proof} Because $\tau$ is faithful, we have $\tau(|x|)=0$ only if $|x|=0$, and so 
$x=0$ when $\tau(|x|)=0$. Moreover, map $x\mapsto\tau(|x|)$ clearly satisfies 
$\alpha x\mapsto |\alpha|\tau(|x|)$.

To prove the triangle inequality, let $x,y\in\A$. By \cite[Theorem 4.2]{akemann--anderson--pedersen1982},
for each $\varepsilon>0$ there exist unitaries $u,v\in \A$ such that $|x+y|\leq u|x|u^*+v|y|v^*+\varepsilon 1$.
Thus, $\tau(|x+y|)\leq \tau(|x|)+\tau(|y|)+\varepsilon\tau(1)$. As this is true for every $\varepsilon>0$, we deduce that
$\tau(|x+y|)\leq \tau(|x|)+\tau(|y|)$.
\end{proof}

\begin{definition} The norm on $\A$ defined in Proposition \ref{tn} is called the \emph{trace norm} on $\A$,
and is denoted by $\|\cdot\|_{1,\tau}$. The induced metric $d_1^\tau(\sigma,\rho)=\|\sigma-\rho\|_{1,\tau}$ on
the density space $\mathcal D_\tau(\A)$ is called the \emph{trace-norm metric}.
\end{definition}

\subsection{The Bures metric is a metric}

To establish the metric property of the Bures distance, we begin with the von Neumann algebra category.

\begin{lemma}\label{vna-lemma1} If $\sigma,\rho \in\mathcal D_\tau(\N)$, then 
\begin{enumerate}
\item $d_B^\tau(\sigma,\rho)=d_B^\tau(\rho,\sigma)$, and
\item if $d_B^\tau(\sigma,\rho)=0$, then $\sigma=\rho$.
\end{enumerate}
\end{lemma}

\begin{proof} For each $x\in\N$, let $\mu_z(t)=\mbox{inf}\,\left\{\|xe\|\,|\,e^*=e=e^2,\,\tau(1-e)\le t\right\}$,
which is a Borel measurable function on $[0, \infty)$ and which, by \cite{fack1982}, satisfies 
\[
\tau(|x|)=\int_0^{\tau(1)}\mu_x(t)\,dt.
\]
Moreover, the results of \cite{fack1982} show that
$\mu_x=\mu_{x^*}=\mu_{|x|}$ and, for any $x,y\in \N$,
the equation $\mu_{|yx^*|}(t)=\mu_{|xy^*|}(t)$ holds for every $t\in[0,\infty)$. 
In addition, if $h\in \N_+$ and if 
$\psi:[0,\infty)\rightarrow[0,\infty)$ is an increasing continuous function
such that $\psi(0)=0$, then 
$\mu_{\psi(h)}(t)=\psi\left(\mu_h(t)\right)$, for all $t\in[0,\infty)$. In particular, with 
$\psi(t)=\sqrt{t}$, we deduce that
$\tau(|\sigma^{1/2}\rho^{1/2}|)=\tau(|\rho^{1/2}\sigma^{1/2}|)$, for all $\sigma,\rho \in\mathcal D_\tau(\N)$.
Hence, $d_B^\tau(\sigma,\rho)=d_B^\tau(\rho,\sigma)$.

The tracial arithmetic-geometric mean inequality states that
\[
\tau\left(|\sigma^{1/2}\rho^{1/2}|\right)\leq \frac{1}{2}\left(\tau(\sigma)+\tau(\rho)\right) =1.
\]
Therefore, $d_B^\tau(\sigma,\rho)=0$ implies that $\tau\left(|\sigma^{1/2}\rho^{1/2}|\right)=\frac{1}{2}\left(\tau(\sigma)+\tau(\rho)\right)$,
which in turn implies that $\sigma=\rho$ by \cite[Theorem 3.4]{farenick--manjegani2005}.
\end{proof}

\begin{lemma}\label{vna-lemma2} If $\tau$ is a faithful normal trace functional on a von Neumann algebra $\N$,
and if $\sigma,\rho\in\mathcal D_\tau(\N)$, then
\begin{enumerate}
\item $\sqrt{2}\,d_B^\tau(\sigma,\rho)\leq \sqrt{ \tau\left( |\sigma^{1/2}-\rho^{1/2}w|^2\right)}$, for every $w\in \N$ of norm $\|w\|\leq1$, and 
\item there exists a unitary $w\in\N$ for which equality in {\rm (1)} holds.
\end{enumerate}
\end{lemma}

\begin{proof} If $w\in\N$ has norm $\|w\|\leq 1$, then $1-w^*w\in \N_+$ and
\[
\begin{array}{rcl}
\tau\left( |\sigma^{1/2}-\rho^{1/2}w|^2\right)&=&2\left(1-\Re\left[\tau(w\sigma^{1/2}\rho^{1/2}\right]\right) \\ && \\
&\geq & 2\left(1-\left|\tau(w\sigma^{1/2}\rho^{1/2}\right|\right) \\ && \\
&\geq& 2\left(1-\tau(|w\sigma^{1/2}\rho^{1/2}|)\right) \\ && \\
&\geq& 2\left(1-\tau(|\sigma^{1/2}\rho^{1/2}|)\right),
\end{array}
\]
where the final inequality follows from $|wx|^2=x^*(w^*w)x\leq x^*x=|x|^2$ and the monotonicity of the square root function $t\mapsto t^{1/2}$
in functional calculus. Thus, $\sqrt{2}\,d_b^\tau(\sigma,\rho)\leq \sqrt{ \tau\left( |\sigma^{1/2}-\rho^{1/2}w|^2\right)}$.

Because $\N$ is a finite von Neumann algebra, there exists an extreme point $v$ of the unit ball of $\N$ such that
$\sigma^{1/2}\rho^{1/2}=v|\sigma^{1/2}\rho^{1/2}|$ \cite{choda1970}. Therefore, since the extreme points of the unit ball of a finite von Neumann algebra
are necessarily unitary, $v^*v=vv^*=1$. Thus,
$v^*\sigma^{1/2}\rho^{1/2}=|\sigma^{1/2}\rho^{1/2}|=|\sigma^{1/2}\rho^{1/2}|^*=\rho^{1/2}\sigma^{1/2}v$. Hence, 
\[
2-2\tau\left(| \sigma^{1/2}\rho^{1/2} |\right)=2-2\Re\left[\tau(v^*\sigma^{1/2}\rho^{1/2})\right] =\tau\left( |\sigma^{1/2}-\rho^{1/2}v^*|^2\right),
\]
which yields equality in {\rm (1)}.
\end{proof}

\begin{lemma}\label{vna-lemma3} If $\sigma,\rho,\theta\in\mathcal D_\tau(\N)$, then
\[
d_B^\tau(\sigma,\rho)\leq d_B^\tau(\sigma, \theta)+d_B^\tau(\theta,\rho).
\]
\end{lemma}

\begin{proof}
By Lemma \ref{vna-lemma2}, there are unitaries $u,w\in\N$ such that
\[
\sqrt{2}\,d_B^\tau(\sigma,\theta)= \sqrt{ \tau\left( |\sigma^{1/2}-\theta^{1/2}w|^2\right)}
\mbox{ and }
\sqrt{2}\,d_B^\tau(\theta,\rho)= \sqrt{ \tau\left( |\theta^{1/2}-\rho^{1/2}u|^2\right)}.
\]
Let $v=uw$. Thus,
\[
\begin{array}{rcl}
\sigma^{1/2}-\rho^{1/2}v&=& \sigma^{1/2}-\theta^{1/2}w + \theta^{1/2}w - \rho^{1/2}v \\&&\\
&=&(\sigma^{1/2}-\theta^{1/2}w)+(\theta^{1/2}-\rho^{1/2}vw^*)w \\&&\\
&=&(\sigma^{1/2}-\theta^{1/2}w)+(\theta^{1/2}-\rho^{1/2}u)w.
\end{array}
\]

Let $x=\sigma^{1/2}-\theta^{1/2}w$ and $y=-(\theta^{1/2}-\rho^{1/2}u)w$ so that $x-y=\sigma^{1/2}-\rho^{1/2}v$.
The Cauchy-Schwarz inequality for the sesquilinear form $(x,y)\mapsto \tau(xy^*)$ yields
$\left| \tau(xy^*)\right| \leq \sqrt{\tau(|x|^2)\,\tau(|y|^2)}$, and so
\[
\begin{array}{rcl}
\left( \sqrt{\tau(|x|^2)}+\sqrt{\tau(|y^2|)}\right)^2 &=&\tau(|x|^2)+\tau(|y|^2)+2\sqrt{\tau(|x|^2)\,\tau(|y|^2)} \\ && \\
&\geq& \tau(|x|^2)+\tau(|y|^2)+2\left| \tau(xy^*)\right| \\ && \\
&\geq& \tau(|x|^2)+\tau(|y|^2)+2\Re\left[ \tau(xy^*)\right] \\ && \\ 
&=&\tau(|x-y|^2).
\end{array}
\]
That is,
\[
\begin{array}{rcl}
\sqrt{ \tau\left(\left|\sigma^{1/2}-\rho^{1/2}v\right|^2\right)} &\leq &
\sqrt{ \tau\left(\left|\sigma^{1/2}-\theta^{1/2}w\right|^2\right)} +
\sqrt{ \tau\left(\left|(\theta^{1/2}-\rho^{1/2}u)w\right|^2\right)} \\ && \\
&=&
\sqrt{ \tau\left(\left|\sigma^{1/2}-\theta^{1/2}w\right|^2\right)} +
\sqrt{ \tau\left(\left| \theta^{1/2}-\rho^{1/2}u \right|^2\right)} \\ && \\
&=& 
\sqrt{2}\,d_B^\tau(\sigma, \theta)+\sqrt{2}\,d_B^\tau(\theta,\rho).
\end{array}
\]
Because Lemma \ref{vna-lemma2} asserts that $\sqrt{2}\,d_B^\tau(\sigma,\rho)\leq \sqrt{ \tau\left(\left|\sigma^{1/2}-\rho^{1/2}v\right|^2\right)}$, 
the triangle inequality $d_B^\tau(\sigma,\rho)\leq d_B^\tau(\sigma, \theta)+d_B^\tau(\theta,\rho)$ follows.
\end{proof}

Lemmas \ref{vna-lemma2} and \ref{vna-lemma3} above are modelled on matrix-theoretic results: see \cite[Exercise 2.20]{Hayashi-book}.

We are now ready to prove that the Bures metric is indeed a metric in the formal sense.

\begin{theorem}\label{bures metric is a metric} 
The function $d_B^\tau$ is a metric on $\mathcal D_\tau(\A)$, for every unital C$^*$-algebra $\A$ 
and faithful trace functional $\tau$ on $\A$.
\end{theorem}

\begin{proof} If $\A$ is a von Neumann algebra and if $\tau$ is a faithful normal trace functional, then Lemmas \ref{vna-lemma1}, 
\ref{vna-lemma2}, and \ref{vna-lemma3} show that $d_B^\tau$ is a metric. If these hypotheses on $\A$ and $\tau$ are not in effect,
then let the GNS representation of $\tau$ be given by $\tau(x)=\langle\pi(x)\xi,\xi\rangle$, for $x\in\A$, where $\pi:\A\rightarrow \B(\H)$
is a unital $*$-homomorphism and $\xi\in \H$ is a unit cyclic vector for the C$^*$-algebra $\pi(\A)$. Let $\N$ denote the double commutant
of $\pi(\A)$. By \cite[Proposition V.3.19]{Takesaki-bookI}, there exists a faithful normal trace on $\N$ such that $\tau=\tau_{\N}\circ\pi$.
Because $d_B^\tau(\sigma,\rho)=d_B^{\tau_{\N}}\left(\pi(\sigma),\pi(\rho)\right)$, for all 
$\sigma,\rho\in\mathcal D_\tau(\A)$, the metric properties of $d_B^\tau$ are inherited from the metric properties of $d_B^{\tau_{\N}}$.
\end{proof}

Fidelity satisfies $F_\tau(\sigma,\rho)\in[0,1]$, for all $\sigma,\rho\in \mathcal D_\tau(\A)$,
and thus the values of the Bures metric $d_b^\tau$ lie in the closed interval $[0,1]$.
Furthermore, a pair of $\tau$-density elements $\sigma$ and $\rho$ are at maximal distance apart 
if and only if $\sigma\rho=\rho\sigma=0$ \cite[Theorem 2.4]{farenick--jaques--rahaman2016}, 
which is an algebraic orthogonality relation that we shall denote by $\sigma\bot\rho$. More generally,
two elements $x,y\in\mathcal T$ are \emph{orthogonal}, denoted by $x\bot y$, if $xy=yx=x^*y=xy^*=0$.

The relationship between the metric induced by the trace norm and the Bures metric is made clear by the following fundamental
inequality \cite{fuchs--vandegraaf1999} for fidelity. (The proof in \cite{fuchs--vandegraaf1999} is for the case of matrices.)

\begin{proposition}[Fuchs-van de Graaf Inequality]\label{fuchs} 
For all $\sigma,\rho\in\mathcal D_\tau(\A)$,
\[
1-\frac{1}{2}\|\rho-\sigma\|_{1,\tau}\leq F_\tau(\rho,\sigma)\leq \sqrt{1-\frac{1}{4}\|\rho-\sigma\|_{1,\tau}^2}.
\]
Equivalently,
\[
2-2F_\tau(\rho,\sigma)\leq \|\rho-\sigma\|_{1,\tau}\leq 2\sqrt{1-F_\tau(\rho,\sigma)^2}.
\]
\end{proposition}
\begin{proof} As in the proof of Theorem \ref{bures metric is a metric}, we may assume without loss of
generality that $\A$ is a unital C$^*$-subalgebra of a finite von Neumann algebra $\N$ with faithful normal trace $\tau$
such that $\A$ is dense in $\N$ with respect to the strong operator topology.

The Powers-St{\o}rmer inequality \cite{hoa--osaka--toan2013,powers--stormer1970} 
asserts that $2\tau(a^{1/2}b^{1/2})\geq \tau(a+b-|a-b|)$, for all $a,b\in\A_+$. Because
$|\tau(x)|\leq \tau(|x|)$ for every $x\in\A$, we obtain
$2-\|\rho-\sigma\|_{1,\tau}\leq 2F_\tau(\rho,\sigma)$, which gives the first inequality.

For the second inequality, observe that, for any unitary $u\in\N$ and $\rho,\sigma\in\mathcal D_\tau(\A)$,
\[
 2(\rho-\sigma)=(\rho^{1/2}+\sigma^{1/2}u^*)(\rho^{1/2}-u\sigma^{1/2})+(\rho^{1/2}-\sigma^{1/2}u^*)(\rho^{1/2}+u\sigma^{1/2}).
\]
By the triangle inequality for the trace norm and using the H\"older inequality \cite{fack-kosaki}, we see that 
\[
\|(\rho-\sigma)\|_{1,\tau}^2\leq\|(\rho^{1/2}+\sigma^{1/2}u^*)\|_{2,\tau}^2\|(\rho^{1/2}-u\sigma^{1/2})\|_{2,\tau}^2,
\]
where $\|x\|_{2,\tau}$ is given by $\sqrt{\tau(x^*x)}$.
Simplifying the right hand side yields
\[
\begin{array}{rcl}
\|(\rho-\sigma)\|_{1,\tau}^2 &\leq& (2+2  \Re \  \tau(\rho^{1/2}\sigma^{1/2}u))(2-2  \Re \  \tau(\rho^{1/2}\sigma^{1/2}u)) \\ && \\
&=&4-4( \Re \  \tau(\rho^{1/2}\sigma^{1/2}u))^2.
\end{array}
\]
By Lemma \ref{vna-lemma2}, $F_\tau(\rho,\sigma)=\displaystyle\sup_{u\in \mathcal U(\N)}\Re \ \tau(\rho^{1/2}\sigma^{1/2}u)$,
where $\mathcal U(\N)$ is the unitary group of $\N$. Therefore, the inequalities above imply that
$4F(\rho,\sigma)^2\leq 4-\|(\rho-\sigma)\|_{1,\tau}^2$. 
\end{proof}

If $d_1^\tau$ denotes the metric on $\A$ given by $d_1^\tau(x,y)=\|x-y\|_{1,\tau}$, then we have:

\begin{proposition}\label{top equiv} The metric spaces $(\mathcal D_\tau(\A),d_B^\tau)$ and $(\mathcal D_\tau(\A),d_1^\tau)$
are homeomorphic.
\end{proposition}

\subsection{Joint concavity of fidelity}

The joint concavity of fidelity, which is well known for matrices, extends to the level of C$^*$-algebras as well.

\begin{lemma}\label{trace concavity} If $a,b,c,d \in \A_+$, then
\[
\tau(|(a+b)^{1/2}(c+d)^{1/2}|) \geq  \tau(|a^{1/2}c^{1/2}|)+ \tau(|b^{1/2}d^{1/2}|).
\]
\end{lemma}

\begin{proof}
By the variational principle of \cite[Proposition 2.2]{farenick--jaques--rahaman2016}), if $a,b \in \A_+$, then
\[
\tau(|a^{1/2}b^{1/2}|)=\frac{1}{2}\inf_{y\in\lgGL(\A)_+}\left(\tau(a y)+\tau(by^{-1})\right).
\]
Therefore, if $a,b,c,d \in\A_+$, then
\begin{align*}
\tau(|(a+b)^{1/2}(c+d)^{1/2}|)&=\frac{1}{2}\inf_{y\in\lgGL(\A)_+}\left(\tau((a+b) y)+\tau((c+d) y^{-1})\right)\\
&=\frac{1}{2}\inf_{y\in\lgG(\A)_+}\left([\tau((ay)+\tau(cy^{-1})]+[\tau(by)+ \tau(dy^{-1})]\right)\\
&\geq \frac{1}{2}\inf_{y\in\lgGL(\A)_+}\left([\tau((ay)+\tau(cy^{-1})]\right) \\
&\qquad\qquad\qquad+\frac{1}{2}\inf_{y\in\lgGL(\A)_+}\left([\tau(by) +\tau(dy^{-1})]\right)\\
&=\tau(|a^{1/2}c^{1/2}|)+\tau(|b^{1/2}d^{1/2}|).
\end{align*}
The inequality above
occurs from the fact that if $f,g$ are two positive functions, then $\inf_{x}(f+g)(x)\geq \inf_{x}(f(x))+\inf_{x}(g(x))$.
\end{proof}

\begin{proposition}[Joint Concavity of Fidelity]\label{joint concavity} If $\sigma_j,\rho_j \in \mathcal D_\tau(\A)$, 
for $j=1,2$, and if $\lambda\in[0,1]$, then
\[
F_\tau(\lambda\sigma_1+(1-\lambda)\sigma_2,\lambda\rho_1+(1-\lambda)\rho_2)\geq 
\lambda F_\tau(\sigma_1,\rho_1)+(1-\lambda)F_\tau(\sigma_2,\rho_2).
\]
\end{proposition}

\begin{proof} Invoke Lemma \ref{trace concavity}.
\end{proof}

\section{Bures Contractive Channels}

\subsection{Definition and examples}

\begin{definition} Suppose that $f:X\rightarrow X$ is a mapping on a metric space $(X,d)$. Then:
\begin{enumerate}
\item $f$ is \emph{nonexpansive}, if $d\left(f(x_1),f(x_2)\right)\leq d(x_1,x_2)$, for all $x_1,x_2\in X$;
\item $f$ is \emph{locally contractive}, if $d\left(f(x_1),f(x_2)\right) < d(x_1,x_2)$, for all distinct $x_1,x_2\in X$;
\item $f$ is \emph{strictly contractive}, if there exists a constant $0\leq C<1$ such that $d\left(f(x_1),f(x_2)\right)\leq Cd(x_1,x_2)$, for all $x_1,x_2\in X$.
\end{enumerate}
\end{definition}

Of course, our interest in this paper is with the metric space $\left(\mathcal D_\tau(\A), d_B^\tau\right)$.

\begin{proposition}\label{nonexpansive} If $\Ep:\A\rightarrow\A$ is a channel, then the function $f_{\Ep}:\mathcal D_\tau(\A)\rightarrow\mathcal D_\tau(\A)$
defined by $f_{\Ep}(\rho)=\Ep(\rho)$, for $\rho\in\mathcal D_\tau(\A)$, is a nonexpansive continuous affine function on the convex metric space $\left(\mathcal D_\tau(\A), d_B^\tau\right)$.
\end{proposition}

\begin{proof} Fidelity satisfies 
$F_\tau(\sigma,\rho)\leq F_\tau\left(\Ep(\sigma),\Ep(\rho)\right) $, for all $\sigma,\rho\in \mathcal D_\tau(\A)$  \cite[Theorem 2.4]{farenick--jaques--rahaman2016}. 
Thus, $d_B^\tau\left(\Ep(\sigma),\Ep(\rho)\right) \leq d_B^\tau(\sigma,\rho)$, which implies that $f_{\Ep}$ is a nonexpansive map of the metric space 
$\mathcal D_\tau(\A)$, and the continuity of $f_{\Ep}$ follows immediately because nonexpansive maps of metric spaces are continuous.
The map $f_{\Ep}$ is obviously affine because $\Ep$ is linear.
\end{proof}

\begin{definition} A channel $\Ep:\A\rightarrow\A$ is a \emph{Bures contraction} if $f_{\Ep}$ is a locally contractive map of the metric space
$\left(\mathcal D_\tau(\A), d_B^\tau\right)$.
\end{definition} 

\begin{example}\label{eg:depolarising} The \emph{completely depolarising channel} $\Omega :\A\rightarrow\A$ defined by
\[
\Omega (x)=\frac{\tau(x)}{\tau(1)}1,
\]
for $x\in \A$, is Bures contractive completely positive channel.
\end{example}

\begin{proof} Because $\Omega $ maps the set $\mathcal D_\tau(\A)$ to the singleton set $\{\frac{\tau(x)}{\tau(1)}1\}$, 
the map $\Omega $ satisfies $d_B^\tau(\Omega (\sigma),\Omega (\rho))=0$ for all $\sigma,\rho\in \mathcal D_\tau(\A)$.
The complete positivity of $\Omega $ follows from the fact that the range of $\Omega $ is the
abelian C$^*$-algebra $\mathbb C\,1$.
\end{proof}

A unitary channel $x\rightarrow uxu^*$ is an isometry of the density space $\mathcal D_\tau(\A)$, and so such channels are not Bures contractive.
Random unitary channels may also fail to be Bures contractive.

\begin{example}\label{ru} The channel $\Ep_\lambda:\M_2(\mathbb C)\rightarrow\M_2(\mathbb C)$ defined, for  $x\in\M_2(\mathbb C)$, by
$\Ep_\lambda(x)=\lambda uxu^* + (1-\lambda)vxv^*$, where  $\lambda\in[0,1]$ and
\[
u=2^{-1/2}\left[ \begin{array}{cc} 0&1 \\ 1&0 \end{array}\right]  \;\mbox{\rm and }\;
v=2^{-1/2}\left[ \begin{array}{cc} 0&-i \\ i&0 \end{array}\right] ,
\]
is not a Bures contraction.
\end{example}

\begin{proof} With respect to the normalised trace $\tau$
on $\M_2(\mathbb C)$, the matrix units $e_{11}$ and $e_{22}$ are density elements. 
If $\lambda$ is neither $0$ nor $1$, then $\Ep_\lambda (e_{11})=e_{22}$ and $\Ep_\lambda (e_{22})=e_{11}$.
Hence,
\[
d_B^\tau\left(\Ep_\lambda (e_{11}),\Ep_\lambda (e_{22}) \right) = d_B^\tau\left( e_{22}, e_{11} \right)=d_B^\tau\left( e_{11}, e_{22} \right),
\]
which shows that $\Ep_\lambda$ is not a Bures contraction. If $\lambda$ is $0$ or $1$, then the channel is a unitary channel, failing again
to be Bures contractive.
\end{proof}

\begin{proposition}\label{add Bures contraction} If $\Ep_1$ and $\Ep_2$ are channels on $\A$, and if at
least one of them is Bures contractive, then so is $\lambda \Ep_1+(1-\lambda)\Ep_2$, for every $\lambda\in(0,1)$.
\end{proposition}

\begin{proof} Without loss of generality, assume that $\Ep_1$ is Bures contractive, and let $\sigma,\rho\in\mathcal D_\tau(\A)$.
By Corollary \ref{joint concavity},
\[
F_\tau(\lambda\Ep_1(\sigma)+(1-\lambda)\Ep_2(\sigma),\lambda\Ep_1(\rho)+(1-\lambda)\Ep_2(\rho))
\]
\[
\begin{array}{rcl}
  &\geq &
\lambda F_\tau(\Ep_1(\sigma),\Ep_1(\rho))+(1-\lambda)F_\tau(\Ep_2(\sigma),\Ep_2(\rho)) \\  &&\\
&>&\lambda F_\tau(\sigma,\rho)+(1-\lambda)F_\tau(\sigma,\rho) \\  &&\\
&=& F_\tau(\sigma,\rho).
\end{array}
\]
Thus, $\lambda \Ep_1+(1-\lambda)\Ep_2$ is Bures contractive.
\end{proof}
 
\begin{corollary}\label{referee suggestion}
The relative interior of the convex set of channels consists entirely of Bures contractive channels. In particular, if 
$\Ep:\A\rightarrow\A$ is a channel, then for every $\varepsilon>0$
there exists Bures contractive channels $\Ep'$ and $\Ep''$ such that
$ \|\Ep-\Ep'\|<\varepsilon$ and $d_B^\tau\left(\Ep(\rho), \Ep''(\rho) \right) < \varepsilon$, for all $\rho\in\mathcal D_\tau(\A)$.
\end{corollary}

\begin{proof} Proposition \ref{add Bures contraction} implies that every channel in the 
relative interior of the convex set of channels is a Bures contractive channel. Because
the relative interior of a convex set is dense in the convex set, the remaining assertions are immediate.
\end{proof}

Proposition \ref{add Bures contraction} gives an efficient construction of Bures contractive channels.

\begin{example}[Depolarising Channels] For each $\lambda\in(0,1)$, the depolarising channel $\Ep_\lambda:\A\rightarrow\A$,
defined by $\Ep_\lambda(x)=\lambda x +(1-\lambda)\frac{\tau(x)}{\tau(1)}1$, is Bures contractive.
\end{example}

\begin{proof} $\Ep_\lambda$ is a convex combination of the identity channel $I$ and the completely depolarising (Bures contractive)
channel $\Omega$. Therefore, by Proposition \ref{add Bures contraction}, $\Ep_\lambda$ is a Bures contraction.
\end{proof}

\begin{example}\label{s but not 2-pos} The channel $\Ep:\M_2(\mathbb C)\rightarrow\M_2(\mathbb C)$ defined by
\[
\Ep\left(\left[ \begin{array}{cc} x_{11}& x_{12} \\ x_{21} & x_{22} \end{array}\right]\right)
=\frac{1}{2}\left[\begin{array}{cc} x_{11}+\frac{ x_{11}+x_{22}}{2} & x_{21} \\ x_{12} & x_{22}+\frac{ x_{11}+x_{22}}{2} \end{array}\right]
\]
is a (non $2$-positive)  Bures contractive Schwarz channel with respect to the canonical trace $\tr$ of $\M_2(\mathbb C)$.
\end{example}

\begin{proof} The proofs that $\Ep$ is a Schwarz map and fails to be $2$-positive are given in \cite{choi1980b}. It is clear that $\Ep$ is trace preserving;
hence, $\Ep$ is a Schwarz channel. Let $\Ep_1$ and $\Ep_2$ be given by
\[
\Ep_1(x)=\frac{\tr(x)}{2}\,1 \;\mbox{ and }\; 
\Ep_2(x)=\frac{1}{2}x^t,
\]
where $x^t$ denotes the transpose of $x$.  Since the map $\Ep_1$ is the completely depolarising channel on $\M_2(\mathbb C)$
with respect to the canonical trace, $\Ep_1$ is a Bures contraction. Therefore, by Proposition \ref{add Bures contraction},  
$\frac{1}{2}(\Ep_1+\Ep_2)=\Ep$ is also a Bures contraction.
\end{proof}

Unitary channels are isometries of the metric space $\left( \mathcal D_\tau(\A), d_B^\tau\right)$. The next result shows that
channels that commute with the unitary channels are necessarily locally contractive maps of $\left( \mathcal D_\tau(\A), d_B^\tau\right)$.

\begin{proposition}\label{covariance}  Assume that a channel $\Ep$ on a finite factor $\N$ commutes with every unitary channel on $\N$; that is,
$\Ep\circ\mbox{\rm Ad}_u = \mbox{\rm Ad}_u \circ \Ep$ for every unitary $u\in\N$, where $\mbox{\rm Ad}_u $ denotes the channel $x\mapsto uxu^*$.
Then there exist nonnegative real numbers $\alpha$ and $\beta$ such that $\Ep(x)=\alpha x+\beta\tau(x)1$, for all $x\in\N$. In particular, 
if $\Ep$ is not a scalar multiple of the identity channel, then $\Ep$ is Bures contractive.
\end{proposition}

\begin{proof}  Let $\tau$ denote the faithful normal trace on $\N$, which without loss of generality we assume to be normalised. Because $\N$ is a factor, $\tau$
is the unique tracial state on $\N$ for which $\tau(1)=1$.

We begin by adapting an argument used in the proof of \cite[Theorem 3.1]{bhat} to our purpose. 
Let $\M$ be any (unital) von Neumann subalgebra of $\N$, and suppose that a unitary $u\in\N$ satisfies $uy=yu$ for every $y\in\M$. 
Therefore, for each $y\in\M$, $\Ep(y)=\Ep(uyu^*)=u\Ep(y)u^*$, by the hypothesis that $\Ep\circ\mbox{\rm Ad}_u = \mbox{\rm Ad}_u \circ \Ep$.
Hence, $\Ep(y)$ commutes with $u$, which implies that
\[
\begin{array}{rcl}
\Ep(\M)&\subseteq& \left(\mbox{Span}\left\{ u\in \M'\cap \N\,:\, u^*u=uu^*=1\right\}\right)'\cap\N \\
&=&(\M'\cap\N)'\cap\N\subseteq\M''=\M.
\end{array}
\]
The property that $\Ep(\M)\subseteq\M$, for every von Neumann subalgebra $\M$ of $\N$, is equivalent to the assertion that
$\Ep$ has the form $\Ep(x)=\alpha x +  \psi(x)1$ for some $\alpha \in\mathbb C$ and linear functional $\psi$ on $\N$ \cite[Theorem 2.1]{bhat}.
The hypothesis $\Ep\circ\mbox{\rm Ad}_u = \mbox{\rm Ad}_u \circ \Ep$ for every unitary $u\in\N$ implies that
$\psi(uxu^*)=\psi(x)$ for all unitaries $u$ and all $x\in\N$. Hence, $\psi$ is a tracial functional. Since $\N$ is a factor, $\psi= \beta \tau$ 
for some $\beta\in\mathbb C$. The trace preservation of $\Ep$ yields $\alpha+\beta=1$, and so $\Ep$ is unital. Hence $\|\Ep\|=1$; thus, if $x\in\ker\tau$ is nonzero, then
$\|x\|\geq\|\Ep(x)\|=|\alpha|\,\|x\|$ and  $|\alpha|\leq1$. Moreover, $\Ep(x^*)=\Ep(x)^*$ for $x\in\ker\tau$ yields $\alpha\in\mathbb R$.
Select a noninvertible positive $h\in\N$; the positivity of $\Ep(h)$ implies that $\alpha\lambda+\beta \geq0$ for every $\lambda$ in the spectrum of $h$.
With $\lambda=0$, in particular, we obtain $\beta\geq 0$. Therefore, $\alpha=1-\beta$ and $|\alpha|\leq 1$ imply that $\alpha,\beta\in[0,1]$.

If $\Ep$ is not a scalar multiple of the identity channel, then $\beta\not=0$.
Thus, $\Ep$ is a convex combination of two channels, one of which is depolarising. Hence, by
Proposition \ref{add Bures contraction}, $\Ep$ is a Bures contraction.
 \end{proof}

\subsection{Multiplicative domains}

\begin{definition} The \emph{multiplicative domain} of a positive linear map
$\Phi:\A\rightarrow\A$ is the set
\[
\mathcal{M}_{\Phi}= \{x\in \A\,:\,  \Phi(xy)=\Phi(x)\Phi(y) ,\,  \Phi(yx)=\Phi(y)\Phi(x), \;\forall \,y \in \A\}.
\]
\end{definition}

Thus, the multiplicative domain of $\Phi$ is the largest C$^*$-subalgebra of $\A$ upon which the linear map $\Phi$ is
multiplicative. A related set is
\[
\mathcal{S}_\Phi=\{x\in\A\,:\,\Phi(x^*x)=\Phi(x)^*\Phi(x),\,
\Phi(xx^*)=\Phi(x)\Phi(x)^*\}.
\]
Of course $\mathcal{S}_\Phi\subseteq \mathcal{M}_{\Phi}$. However, if $\Phi$ is a Schwarz map, then 
$\mathcal{S}_\Phi=\mathcal{M}_{\Phi}$ \cite{choi1974,Stormer-book}.

\begin{proposition}\label{miza1} If $\Ep:\A\rightarrow\A$ is a  
channel, then 
$f_\Ep$ is isometric on the set $\mathcal D_\tau(\A)\cap \mathcal M_\Ep$.
That is,
\[
d_B^\tau\left(\Ep(\sigma),\Ep(\rho)\right)=d_B^\tau(\sigma,\rho),
\]
for all $\sigma,\rho\in \mathcal D_\tau(\A)\cap \mathcal M_\Ep$.
\end{proposition}

\begin{proof}
We shall prove that $F_\tau(\rho,\sigma)=F_\tau(\mathcal{E}(\rho),\mathcal{E}(\sigma))$ 
for all $\sigma,\rho\in \mathcal D_\tau(\A)\cap \mathcal M_\Ep$. 

To this end, select $x\in \A_+\cap \mathcal M_\Ep$. Because $\mathcal M_\Ep$ is a C$^*$-algebra,
the element $x^{1/2}$ also belongs to $\mathcal M_\Ep$. Thus, 
\[
\mathcal{E}(x)=\mathcal{E}(x^{1/2}x^{1/2}) 
=\mathcal{E}(x^{1/2})\mathcal{E}(x^{1/2}) 
=[\mathcal{E}(x^{1/2})]^2,
\]
which shows that $\mathcal{E}(x)^{1/2}=\mathcal{E}(x^{1/2})$.
Now for $\rho,\sigma \in \mathcal D_\tau(\A)\cap \mathcal M_\Ep$, the element
$\rho^{1/2}\sigma\rho^{1/2}$ lies in $\mathcal{M}_{\mathcal{E}}$; thus,
\begin{align*}
F_\tau(\mathcal{E}(\rho),\mathcal{E}(\sigma))&=\tau[(\mathcal{E}(\rho)^{1/2}\mathcal{E}(\sigma)\mathcal{E}(\rho)^{1/2})^{1/2}]\\
&=\tau[(\mathcal{E}(\rho^{1/2})\mathcal{E}(\sigma)\mathcal{E}(\rho^{1/2}))^{1/2}]\\
&=\tau[(\mathcal{E}(\rho^{1/2}\sigma\rho^{1/2}))^{1/2}]\\
&=\tau[\mathcal{E}((\rho^{1/2}\sigma\rho^{1/2})^{1/2})]\\
&=\tau[(\rho^{1/2}\sigma\rho^{1/2})^{1/2}]  \\
&=F_\tau(\rho,\sigma).
\end{align*}
Hence, the affine function $f_\Ep$ is isometric on
$\mathcal D_\tau(\A)\cap \mathcal M_\Ep$.
\end{proof}

\begin{corollary}\label{trivial mult.dom} The multiplicative domain of a Bures contractive channel is $\mathbb C\,1=\{\lambda\,1\,:\,\lambda\in\mathbb C\}$.
\end{corollary}

\begin{proof} If not, then the Bures contractive channel $\mathcal E$ in question possesses at least two distinct $\tau$-density elements
$\sigma$ and $\rho$. Therefore, the feature $d_B^\tau(\Ep(\sigma),\Ep(\rho))<d_B^\tau(\sigma,\rho)$ cannot hold, by 
Proposition \ref{miza1}.
\end{proof}

The preceding result  is a stronger form of a corresponding theorem of 
Raginsky \cite{raginsky2002} for the trace-norm metric.

A partial converse to Proposition \ref{miza1} is the following result on the hereditary nature of multiplicative domains.

\begin{proposition}\label{miza2}
If $\Ep$ is a Schwarz channel, and if $\sigma\in \mathcal M_\Ep$ is a $\tau$-density element for which there exists
$\rho\in\mathcal D_\tau(\A)$ such that $d_B^\tau(\sigma,\rho)=d_B^\tau\left(\Ep(\sigma),\Ep(\rho)\right)$, then 
$\sigma^{1/2}\rho\sigma^{1/2}\in\mathcal M_\Ep$.
\end{proposition}

\begin{proof}
Because $\sigma\in \mathcal{M}_{\mathcal{E}}$, 
we have that $\mathcal{E}(\sigma)^{1/2}=\mathcal{E}(\sigma^{1/2})$ and, for every $x\in \A$, that
$\mathcal{E}(\sigma^{1/2}x)=\mathcal{E}(\sigma^{1/2})\mathcal{E}(x)$. Thus,
using the Schwarz inequality 
\[
\mathcal{E}(\sigma^{1/2}\rho\sigma^{1/2})^{1/2}\geq \mathcal{E}([\sigma^{1/2}\rho\sigma^{1/2}]^{1/2}),
\]
we obtain
\[
\begin{array}{rcl}
\tau([\mathcal{E}(\sigma)^{1/2}\mathcal{E}(\rho)\mathcal{E}(\sigma)^{1/2}]^{1/2}) 
&=&\tau([\mathcal{E}(\sigma^{1/2})\mathcal{E}(\rho)\mathcal{E}(\sigma^{1/2})]^{1/2}) \\ && \\
&=&\tau([\mathcal{E}(\sigma^{1/2}\rho \sigma^{1/2})]^{1/2}) \\ && \\
&\geq&\tau\circ\mathcal{E}[(\sigma^{1/2}\rho\sigma^{1/2})^{1/2}] .
\end{array}
\]
Hence, the inequality must be an equality,  
resulting in
\[
[\mathcal{E}(\sigma)^{1/2}\mathcal{E}(\rho)\mathcal{E}(\sigma)^{1/2}]^{1/2} =
\mathcal{E}[(\sigma^{1/2}\rho\sigma^{1/2})^{1/2}].
\]
The equation above implies that 
$\sigma^{1/2}\rho\sigma^{1/2}\in\mathcal M_\Ep$.
\end{proof}

\begin{corollary}
If $\Ep$ is a Schwarz channel, and if $\sigma\in \mathcal M_\Ep$ is an invertible $\tau$-density element for which there exists
$\rho\in\mathcal D_\tau(\A)$ such that $d_B^\tau(\sigma,\rho)=d_B^\tau\left(\Ep(\sigma),\Ep(\rho)\right)$, then $\rho\in \mathcal M_\Ep$.
\end{corollary}
\begin{proof} The hypothesis and Proposition \ref{miza2} imply that
$\sigma^{1/2}\rho\sigma^{1/2}\in\mathcal M_\Ep$. Because $\sigma^{-1/2}\in\mathcal M_\Ep$, the element
$\rho=\sigma^{-1/2}\left(\sigma^{1/2}\rho\sigma^{1/2}\right)\sigma^{-1/2}$ also lies
in $\mathcal M_\Ep$.
\end{proof}

\begin{definition} The \emph{centre} of the convex set $\mathcal D_\tau(\A)$ is the element $\zeta=\tau(1)^{-1}1$.
\end{definition}

The next proposition gives a necessary and sufficient criteria for a positive element to be in the multiplicative domain.

\begin{proposition}{\label{multiplicative dom. lemma}}
If $\Ep:\A\rightarrow\A$ is a Schwarz channel and if $a\in\A_+$ is nonzero, 
then $a\in\mathcal M_\Ep$ if and
only if the $\tau$-density elements
$\tau(a)^{-1}a$ and $\tau(a)^{-1}\Ep(a)$ are equidistant from the centre $\zeta$ of $\mathcal D_\tau(\A)$.
\end{proposition}

\begin{proof} 
Assume that $a\in\mathcal M_\Ep$ and let $\rho=\tau(a)^{-1}a$. 
It was shown in the proof of Proposition \ref{miza1}
that $\mathcal{E}(\rho^{1/2})=\mathcal{E}(\rho)^{1/2}$; therefore,
\[
\begin{array}{rcl}
F_\tau(\rho,\zeta)&=&F_\tau(\rho,\tau(1)^{-1}1) =\displaystyle\frac{\tau(\rho^{1/2})}{\tau(1)^{\frac{1}{2}}} 
=\displaystyle\frac{\tau[\mathcal{E}(\rho^{1/2})]}{\tau(1)^{\frac{1}{2}}}  
=\displaystyle\frac{\tau[\mathcal{E}(\rho)^{1/2}]}{\tau(1)^{\frac{1}{2}}} \\ && \\
&=&F_\tau(\Ep(\rho),\tau(1)^{-1}1) \\ && \\
&=&F_\tau(\Ep(\rho),\zeta).
\end{array}
\]
Hence, 
$\tau(a)^{-1}a$ and $\tau(a)^{-1}\Ep(a)$ are equidistant from the centre $\zeta$ of $\mathcal D_\tau(\A)$.

Conversely, assume that $d_B^\tau\left(\tau(a)^{-1}a,\zeta)\right)=d_B^\tau\left(\tau(a)^{-1}\Ep(a),\zeta)\right)$.
Therefore, the fidelities $F_\tau(\left(\tau(a)^{-1}a,\tau(1)^{-1}1)\right)$ and $F_\tau(\left(\tau(a)^{-1}\Ep(a),\tau(1)^{-1}1)\right)$
coincide, which implies that $\tau(a^{1/2})=\tau\left(\Ep(a)\right)^{1/2}$. 
The Schwarz inequality asserts that
$\mathcal{E}(a)=\mathcal{E}(a^{1/2}a^{1/2})\geq (\mathcal{E}(a^{1/2}))^2$.
Since the square root is an operator monotone function, we obtain
$\mathcal{E}(a)^{1/2}\geq \mathcal{E}(a^{1/2})$.
Hence, 
\[
\tau(a^{1/2})=\tau[\mathcal{E}(a)^{1/2}]
\geq \tau[\mathcal{E}(a^{1/2})]
=\tau(a^{1/2}).
\]
By the faithfulness of trace, 
$\mathcal{E}(a)^{1/2}=\mathcal{E}(a^{1/2})$.
Therefore,
\[
\mathcal{E}(a)=[\mathcal{E}(a)^{1/2}]^{2}
 =[\mathcal{E}(a^{1/2})]^2
 =\mathcal{E}(a^{1/2})\mathcal{E}(a^{1/2}),
\]
which shows that $a^{1/2}\in\mathcal S_\Ep$. Because 
$\mathcal{S}_\Ep=\mathcal{M}_{\Ep}$ for Schwarz maps, we deduce that $a\in\mathcal M_\Ep$.
\end{proof}

Recall that $2$-positive linear maps are Schwarz maps; thus, by Proposition \ref{basic positive facts},
$2$-positive channels are unital maps.

\begin{proposition} Suppose that a $2$-positive channel $\Ep$ has the form 
$\mathcal{E}= \frac{1}{2}(\Phi+\Psi)$, for some $2$-positive channels $\Phi$ and $\Psi$.
If one of $\Phi$ or $\Psi$ is a Bures contraction, then the multiplicative domain of $\Ep$ is $\mathbb C 1$.
\end{proposition}

\begin{proof} By \cite[Theorem 3.3]{choi1974}, 
\[
 \mathcal{M}_{\mathcal{E}}=\mathcal{M}_{\Phi}\cap
\mathcal{M}_{\Psi}\cap\{x\in\A\,:\, \mathcal{E}(x)=\Phi(x)=\Psi(x)\} .
\]
If $\Phi$ or $\Psi$ is a Bures contraction, then it is has multiplicative domain $\mathbb C 1$,
by Corollary \ref{trivial mult.dom}. Therefore, $\mathcal M_\Ep=\mathbb C 1$.
\end{proof}

Because the set of $2$-positive channels is convex, the proposition above
says that any map lying on a line segment passing through a Bures contraction 
will have trivial multiplicative domain.

\subsection{Fixed Points}

\begin{definition} If $\Phi$ is a linear transformation on a vector space $V$, then $x \in V$ is a 
\emph{fixed point} of $\Phi$ if $\Phi(x)=x$. The vector subspace
\[
\mbox{\rm Fix}\,\Phi=\{x\in V\,:\,\Phi(x)=x\} 
\]
is called the \emph{fixed point space of $\Phi$}.
\end{definition}

Our interest here is with the fixed points of positive linear maps and channels.
The first assertion of the following proposition is
widely known, while the second assertion is an algebraic variant of a theorem of Kribs \cite{kribs2003}.
 
The notation $[x,y]$ below denotes the commutator $[x,y]=xy-yx$.

\begin{proposition}\label{kribs1} The following statements hold for a unital
channel $\Ep$ on $\A$:
\begin{enumerate}
\item if $\Ep$ is a Schwarz channel, then $\mbox{\rm Fix}\,\Ep$ is a unital C$^*$-algebra;
\item if there exist $w_1,\dots, w_n\in\A$ such that $\Ep(x)=\displaystyle\sum_{k=1}^n w_kx w_k^*$, for all $x\in \A$,
and if the linear span of the projections in the C$^*$-algebra $\mbox{\rm Fix}\,\Ep$ is dense in $\mbox{\rm Fix}\,\Ep$, then
\[
\mbox{\rm Fix}\,\Ep=\{x\in \A\,:\,[x,w_k]=[x,w_k^*]=0, \; \forall\, k=1,\dots,n\}.
\]
\end{enumerate}
\end{proposition}
\begin{proof}
%
Suppose that $\Ep$ is a Schwarz channel.
If $x\in \mbox{\rm Fix}\,\mathcal E$, then the Schwarz inequality yields
$0\leq\mathcal E(x^*x)-\mathcal E(x^*)\mathcal E(x)=\mathcal E(x^*x)-x^*x$,
where the final equality is a consequence of the hypothesis $x\in \mbox{\rm Fix}\,\mathcal E$.
On evaluating the trace, we obtain the inequality
\[
0\leq\tau\left[\mathcal E(x^*x)-\mathcal E(x^*)\mathcal E(x)\right]=\tau(x^*x)-\tau(x^*x)=0.
\]
Thus, the positive element
$\mathcal E(x^*x)-\mathcal E(x^*)\mathcal E(x)$ has zero trace, which yields 
$\mathcal E(x^*x)=\mathcal E(x^*)\mathcal E(x)$. Hence, 
$x\in \mathcal S_\Ep =\mathcal M_\Ep$. 
Therefore, if $x_1,x_2\in \mbox{\rm Fix}\,\mathcal E$,
then $x_1,x_2\in \mathcal M_\Ep$ and so $\mathcal E(x_1x_2)=\mathcal E(x_1)\mathcal E(x_2)=x_1x_2$,
which proves that $x_1x_2\in \mbox{\rm Fix}\,\mathcal E$.

To prove the second assertion, suppose without loss of generality that $\A$ is 
represented faithfully as a unital C$^*$-subalgebra of $\B(\H)$
for some Hilbert space $\H$. The set $\B_\Ep=\{x\in \A\,:\,[x,w_k]=[x,w_k^*]=0, \; \forall\, k=1,\dots,n\}$
is a unital  C$^*$-subalgebra of $\A$ such that
$\B_\Ep \subseteq \mbox{\rm Fix}\,\Ep$. 
Conversely, choose a projection $p\in \mbox{\rm Fix}\,\Ep$. If $\xi\in \H$,
then,
\[
\|p\xi\|^2 = 
\langle p\xi,\xi\rangle= \langle \Ep(p)\xi,\xi\rangle=\sum_{k=1}^n\langle pw_k^*\xi,w_k^*\xi\rangle=
\sum_{k=1}^n\|p w_k^*\xi\|^2 .
\]
Thus, if $\xi\in\ker p$, then $w_k^*\xi\in\ker p$ for every $k$, which proves that $\ker p$ is invariant for each $w_k^*$.
A similar argument, using $1-p\in\mbox{\rm Fix}\,\mathcal E$ in place of $p$ shows that $\mbox{ran}\,p=\ker (1-p)$
is invariant for each $w_k^*$. Thus, $\mbox{ran}\, p$ is a reducing subspace for each $w_k^*$, which proves that
$pw_k=w_kp$ for each $k$. Hence, $p\in \B_\Ep $. Because the linear span of the projections in 
$\mbox{\rm Fix}\,\mathcal E$ is dense in $\mbox{\rm Fix}\,\mathcal E$,
we deduce from the continuity of $\Ep$ that $\mbox{\rm Fix}\,\mathcal E\subseteq \B_\Ep$.
\end{proof}

A useful result concerning fixed points of Bures contractive Schwarz channels is:

\begin{proposition}\label {bcsm fixed} If $\Ep$ is a Bures contractive Schwarz channel, then $\mbox{\rm Fix}\,\Ep=\mathbb C\,1$
\end{proposition} 

\begin{proof} By Proposition \ref{kribs1},  the fixed point subspace
of a Schwarz channel is a C$^*$-algebra; hence, $\mbox{\rm Fix}\,\Ep$ is a subset of the multiplicative
domain $ \mathcal M_{\Ep}$ of $\Ep$.
In addition, $\Ep$ is a Bures contraction; therefore, $\mathcal M_{\Ep}=\mathbb C\,1$  by Corollary \ref{trivial mult.dom}.
\end{proof}

\section{Frobenius Theory of Bures Contractions}

\subsection{Irreducible channels}

The central idea in Frobenius theory is that of irreducibility. This notion appears as a combinatorial concept in matrix theory,
but in more general contexts the notion of irreducibility is related to the absence of invariant faces in a cone. The main
results of this section show that Bures contractive channels are irreducible, in both the C$^*$-algebra and von Neumann algebra frameworks.

\begin{definition} A nonempty subset $F$ of $\A_+$ is a \emph{face} of $\A_+$ if, for $a\in \A_+$ and $b\in F$, the inequality
$a\leq b$ holds only if $a\in F$.
\end{definition}

\begin{definition}\label{defn:irr} Assume that $\Phi:\A\rightarrow\A$ and $\Psi:\M\rightarrow\M$ are positive linear maps of a
C$^*$-algebra $\A$ and a von Neumann algebra $\M$, and assume that $\Psi$ is normal.
\begin{enumerate}
\item If there exists a norm-closed
face $F$ of $\A_+$ different from $\{0\}$ and $\A_+$ such that $\Phi(F)\subseteq F$, 
then $\Phi$ is said to be \emph{reducible}.
\item If there exists a ultraweakly-closed
face $F$ of $\M_+$ different from $\{0\}$ and $\M_+$ such that $\Psi(F)\subseteq F$, 
then $\Psi$ is said to be \emph{reducible}.
\end{enumerate}
A positive linear map on $\A$ or $\M$ that is not reducible is called \emph{irreducible}.
\end{definition}

The precise determination of the norm-closed faces and the ultraweakly-closed faces of the positive cones of, respectively, 
C$^*$-algebras and von Neumann algebras 
is provided in the monograph \cite{Alfsen--Schultz-book1}.

The concept of irreducibility is highly studied in matrix theory. In our setting, a matrix $\Phi$ with nonnegative entries is a
positive linear map on the 
abelian C$^*$-algebra $\A=\mathbb C^d$ relative to the canonical trace, and such a map $\Phi$ is irreducible in the sense of Definition \ref{defn:irr} 
if and only if $\Phi$ is irreducible in the sense of matrix theory (in that the directed adjacency matrix of $\Phi$ is strongly connected).

In noncommutative algebra, a
decomposable map $\Phi:\M_d(\mathbb C)\rightarrow\M_d(\mathbb C)$ of the form
\[
\Phi(x)=\sum_{k=1}^m a_k x a_k^*  +  \sum_{\ell=1}^n b_\ell x^t b_\ell^*,
\]
for $x\in\M_d(\mathbb C)$, is irreducible if and only if the only invariant subspaces
in common for the family of operators $\{a_k^*,b_\ell^*\}_{k,\ell}$ are the trivial subspaces $\{0\}$ and $\mathbb C^d$
\cite{farenick1996}.

The following observation is a variant of \cite[Proposition 1]{farenick1996}.

\begin{lemma}\label{irr1} If $\Phi:\M\rightarrow\M$ is a 
normal positive linear map of norm $\|\Phi\|\leq 1$ on a von Neumann algebra $\M$,
then $\Phi$ is reducible if and only if
there exists a nontrivial projection $p\in\M$ such that $\Phi(p)\leq p$. 
\end{lemma}

\begin{proof} Assume that $\Phi$ is reducible; thus, there is a 
nontrivial ultraweakly-closed face $F$ of $\M_+$ such that $\Phi(F)\subseteq F$.
Therefore, there exists an ultraweakly-closed left ideal $J$ of $\M$ such that $F=J_+$
\cite[Theorem 3.13]{Alfsen--Schultz-book1}. Now, if $1$ were an element of $F$, then 
$1\in J_+\subseteq J$ implies that $x=x(1)\in J$ for every $x\in \M$, and so we would obtain $J=\M$
and $F=\M_+$, which is impossible since $F$ is nontrivial. 

Consider now $F_1=F\cap\{x\in\M\,:\,\|x\|\leq 1\}$, which is a (weakly-closed) face of $\M_+$.
Because both $F$ and the closed unit ball of $\M$ are invariant under $\Phi$, we deduce that $F_1$ is $\Phi$-invariant.
Furthermore, the exists a projection $p\in F_1$ such that $a\leq p$ for all $a\in F_1$ 
\cite[Proposition 3.9]{Alfsen--Schultz-book1}. By the remarks of the previous paragraph, $p\not=1$; and
since $F\not=\{0\}$, $F$ has a nonzero elements $a$ of norm $\|a\|\leq 1$. Thus, $p$ is a nontrivial projection.
Because $p\in F_1$ and $F_1$ is $\Phi$-invariant, we have that $\Phi(p)\in F_1$, whence $\Phi(p)\leq p$.

Conversely, suppose that $p\in\M$ is a nontrivial projection such that $\Phi(p)\leq p$. Let $F=\{a\in\M_+\,:\,a\leq p\}$,
which is a proper ultraweakly closed face of $\M_+$. Further, if $a\in F$, then $a\leq p$ implies that $\Phi(a)\leq\Phi(p)\leq p$,
and so $\Phi(a)\in F$. Therefore, $\Phi$ is reducible.
\end{proof}

\begin{corollary}\label{irr criterion} A contractive normal trace-preserving positive linear 
map $\Ep$ on a finite von Neumann algebra $\N$ is reducible 
if and only if 
there exists a nontrivial projection $p\in\N$ such that $\Ep(p)=p$.
\end{corollary}

\begin{proof} If $\Ep$ is reducible, then Lemma \ref{irr1} shows that $\Ep(p)\leq p$
for some nontrivial projection $p\in\N$.
Because $\tau\circ\Ep=\tau$,  
\[
0\leq \tau\left(p-\Ep(p)\right)=\tau(p)-\tau\circ\Ep(p)=\tau(p)-\tau(p)=0.
\]
Hence, as the trace $\tau$ is faithful, $\Ep(p)=p$. 

Conversely, if $\Ep(p)=p$, then  $\Ep$ is reducible by Lemma \ref{irr1}. 
\end{proof}

\begin{example}\label{ru irr} Suppose, for $\lambda\in(0,1)$, that $\Ep_\lambda:\M_2(\M)\rightarrow\M_2(\M)$ is the channel
$\Ep_\lambda(x)=\lambda uxu^* + (1-\lambda)vxv^*$, for  $x\in\M_2(\M)$, where 
\[
u= \left[ \begin{array}{cc} 0&1 \\ 1&0 \end{array}\right]  \;\mbox{\rm and }\;
v= \left[ \begin{array}{cc} 0&-i \\ i&0 \end{array}\right] ,
\]
and where $\M$ is a finite von Neumann algebra. Then $\Ep_\lambda$  is reducible, except in the case where $\M\cong\mathbb C$.
\end{example} 

\begin{proof} Because $u$ and $v$ are selfadjoint unitaries, Proposition \ref{kribs1} shows that $\mbox{\rm Fix}\,\Ep_\lambda$ is the
commutant of the set $\{u,v\}$. A straightforward calculation shows that the commutant of $\{u,v\}$
is $\left\{\left[\begin{array}{cc} a&0\\ 0&a \end{array}\right] \,:\,a\in\M\right\}$. Hence, there is always nontrivial 
projection $p$ for which $\Ep_\lambda(p)=p$, except in the case where $\M\cong\mathbb C$. Therefore, by Corollary \ref{irr criterion}, 
$\Ep_\lambda$ is irreducible if and only if $\M \cong\mathbb C$.
\end{proof}

Recall that if $\Phi:\A\rightarrow\A$ is a positive linear map on a unital C$^*$-algebra $\A$, then $\Phi^{**}:\A^{**}\rightarrow\A^{**}$
is a normal positive linear map on the enveloping von Neumann algebra $\A^{**}$ of $\A$.

In what follows, if $S\subseteq\A$, then $S_+$ shall denote $S\cap\A_+$.

\begin{lemma}\label{reducible dd} The following statements are equivalent for a positive linear map $\Phi:\A\rightarrow\A$:
\begin{enumerate}
\item $\Phi$ is reducible;
\item $\Phi^{**}$ is reducible.
\end{enumerate}
\end{lemma}

\begin{proof} If $\pi_u:\A\rightarrow\B(\H_u)$ is the universal representation of $\A$, then $\A\cong\pi(\A)\subseteq\pi(\A)''=\A^{**}$;
therefore, assume without loss of generality that $\A$ is represented as a weakly dense unital C$^*$-subalgebra of $\A^{**}$.

If $\Phi^{**}$ is reducible, then there exists a nontrivial projection $p\in\A^{**}$ such that $\Phi^{**}(p)\leq p$.
The set $F=\{a\in\A_+\,:\,a\leq p\}$ is a norm-closed proper face of $\A_+$ (as $p\not\in\{0,1\}$). Furthermore, if $a\in F$, 
then the inequality $a\leq p$ leads to $\Phi(a)=\Phi^{**}(a)\leq\Phi^{**}(p)\leq p$, implying that $\Phi(a)\in F$. Hence,
$\Phi$ is reducible.

Conversely, suppose that $\Phi$ is reducible. Without loss
of generality we may assume that $\Phi$ is a contraction, as the invariance of a proper norm-closed face $F$ of $\A_+$ under $\Phi$ is
independent of the norm of $\Phi$.
Let $J$ be a norm-closed left ideal such that $F=J_+$, and fix an increasing right 
approximate identity for $J$ -- namely, an increasing net $\{e_\lambda\}_\lambda\subset J_+$ such that
$\lim_\lambda\|x-xe_\lambda\|=0$ for every $x\in J$. 
Considered as a subset of $\A^{**}$, the ultraweak closure of $J$
is $\overline{J}^{\sigma{\rm-wk}}=\{zq\,:\,z\in\A^{**}\}$, where 
$q\in\A^{**}$ is the 
least projection such that $aq=a$ for every $a\in J_+$ and is given by
$q=\sup_\lambda e_\lambda$ \cite[Proposition 3.44]{Alfsen--Schultz-book1}.
Since $\|e_\lambda\|\leq\|q\|=1$ and $\Phi$ is contractive, the element $\Phi(e_\lambda)\in F$ has norm no greater than $1$;
thus, $\Phi(e_\lambda) \leq q$ \cite[Proposition 3.9]{Alfsen--Schultz-book1}.
Hence, using the normality of $\Phi^{**}$, we obtain
\[
\Phi^{**}(q)= \Phi^{**}\left(\sup_\lambda e_\lambda\right) = \sup_\lambda \Phi^{**}(e_\lambda) =\sup_\lambda \Phi (e_\lambda)\leq q,
\]
which implies that $\Phi^{**}$ is reducible, by Lemma \ref{irr1}.
\end{proof}

The following two results establish the relationship between Bures contractiveness and the irreducibility of Schwarz channels.

\begin{theorem}\label{bc irr vN} A Bures contractive Schwarz channel on a finite von Neumann algebra is irreducible.
\end{theorem}

\begin{proof} Let $\Ep:\N\rightarrow\N$ be a Bures contractive Schwarz channel on a finite von Neumann algebra $\N$
with faithful normal trace $\tau$. Assume that $\Ep$ is reducible. Thus, there is a nontrivial projection $p$ such that
$\Ep(p)=p$, by Proposition \ref{irr1}. Since $p=pp^*=p^*p=\Ep(p)^*\Ep(p)$ and $\Ep(p)=\Ep(p^*p)$, the projection $p$
is in the multiplicative domain of $\Ep$. Corollary \ref{trivial mult.dom} asserts that the multiplicative domain of a 
Bures contractive Schwarz channel 
is $\mathbb C 1$.
Hence, $p$ must be $0$ or $1$, in contradiction to the fact that $p$ is neither $0$ nor $1$. Therefore, $\Ep$ must be
irreducible.
\end{proof}

\begin{theorem}\label{bc irr C*} A Bures contractive Schwarz channel on a tracial C$^*$-algebra is irreducible.
\end{theorem}

\begin{proof} Let $\Ep:\A\rightarrow\A$ be a Bures contractive Schwarz channel. 
Assuming that $A\subseteq\A^{**}$, the
faithful $\Ep$-invariant trace $\tau$ on $\A$ extends to a faithful $\Ep^{**}$-invariant
trace $\tau^{**}$ on $A^{**}$. 

Assume, contrary to what we aim to prove, that
$\Ep$ is reducible. Therefore, by Lemma \ref{reducible dd},   $\Ep^{**}$ is reducible; thus, $\Ep^{**}(p)=p$, for some nontrivial projection
$p\in\A^{**}$. Likewise, $\Ep^{**}(q)=q$, where $q=1-p$. Hence, 
$p$ and $q$ are nonzero positive contractions $ \A^{**}$ such that
$p\bot q$ and $\Ep^{**}(p)\bot \Ep^{**}(q)$. If $a,b\in\A^{**}$ are positive and satisfy $a \leq p$ and $b\leq q$, then $ab=ba=0$. To verify this, note that
\[
0\leq \tau(b^{1/2}ab^{1/2})\leq \tau^{**}(b^{1/2}pb^{1/2})=\tau^{**}(pbp)\leq\tau^{**}(pqp)=\tau^{**}(pq)=0,
\]
which implies that $b^{1/2}ab^{1/2}=0$. As these are operators acting on the universal representation Hilbert space $\H_u$ for $\A$, we obtain
$0=\|ab^{1/2}\xi\|^2$ for every $\xi\in\H_u$; hence, $ab^{1/2}=0$, which yields $ab=0$ and $ba=0$.

By the Kaplansky Density Theorem, there are increasing nets $\{a_\alpha\}_\alpha$ and $\{b_\beta\}_\beta$ of positive operators
 in $\A$ converging strongly to $p$ and $q$, respectively. Thus, $p=\sup_\alpha a_\alpha$ and $q=\sup_\beta b_\beta$ yields, by the normality
 of $\Ep^{**}$, that 
 \[
 \Ep^{**}(p)=\sup_\alpha\Ep^{**}(a_\alpha)=\sup_\alpha\Ep (a_\alpha) \;\mbox{ and }\; \Ep^{**}(q)=\sup_\beta\Ep^{**}(b_\beta)=\sup_\beta\Ep (b_\beta).
 \]
 Thus, by the arguments of the previous paragraph,
 $a_\alpha \bot b_\beta$ and $\Ep(a_\alpha)\bot\Ep(b_\beta)$ for all $\alpha$ and $\beta$, which contradicts the hypothesis that $\Ep$
 is a Bures contraction.
Therefore, it must be that $\Ep$ is irreducible.
\end{proof}

Examples \ref{ru} and \ref{ru irr} demonstrate that the converses of Propositions \ref{bc irr vN} and 
\ref{bc irr C*} do not hold.

\begin{proposition}\label{cor kribs1} If $\Ep$ is a Schwarz channel on $\N$,
then there exists a (normal) completely positive channel $\Pi:\N\rightarrow\N$
such that $\Pi^2=\Pi$ and $\mbox{\rm ran}\,\Pi=\mbox{\rm Fix}\,\Ep$. Furthermore, 
$\Ep$ is an irreducible channel if and only if $\Pi$ is Bures contractive channel.
\end{proposition} 

\begin{proof} The hypothesis on $\Ep$ implies that $\mbox{\rm Fix}\,\Ep$ is a unital C$^*$-algebra (Proposition \ref{kribs1}).
Furthermore, because $\Ep$ is normal, the 
C$^*$-algebra $\mbox{\rm Fix}\,\Ep$ is a von Neumann subalgebra of $\N$. Hence,
there exists a unital completely positive normal map $\Pi:\N\rightarrow\N$
such that $\Pi^2=\Pi$, $\tau\circ\Pi=\tau$, and $\mbox{\rm ran}\,\Pi=\mbox{\rm Fix}\,\Ep$ \cite[Propositions 2.2.6, 2.2.11]{Stormer-book}.

Suppose that $\Ep$ is irreducible. Thus, $\mbox{\rm Fix}\,\Ep=\mathbb C\,1$, and so $\Pi(x)=\lambda_x1$ for every $x\in\N$. 
From $\tau\circ\Pi=\tau$ we deduce that $\lambda_x=\frac{\tau(x)}{\tau(1)}$. Hence, $\Pi$ is depolarising and, therefore, Bures contractive.

Conversely, assume that $\Pi$ is a Bures contraction. Thus, $\mbox{\rm Fix}\,\Ep=\mathbb C\,1$ implies that $\mbox{\rm Fix}\,\Ep$  has no nontrivial projections.
Therefore, the equation $\Ep(p)=p$ holds for a projection $p\in\N$ only if $p$ is trivial. Hence, $\Ep$ is irreducible.
\end{proof}

\subsection{Perron value}

The identity operator on $\A$ will be denoted by $I$ and the spectrum  
of a bounded linear operator $\Phi:\A\rightarrow\A$ is denoted by  $\mbox{\rm Sp}\,\Phi$.
That is,
\[
\mbox{\rm Sp}\,\Phi=\{\lambda\in\mathbb C\,|\,\Phi-\lambda I \mbox{ is not an invertible operator on }\A\}.
\]
Of special interest are the point spectrum $\mbox{\rm Sp}_{\rm p}\,\Phi$, which consists of the eigenvalues of $\Phi$, 
and the approximate point spectrum $\mbox{\rm Sp}_{\rm ap}\,\Phi$, which consists of approximate eigenvalues of $\Phi$.
Thus, $\lambda\in \mbox{\rm Sp}_{\rm ap}\,\Phi$ if and only for every $\varepsilon>0$ there exists a nonzero $x\in\A$ such that
$\|\Phi(x)-\lambda x\|<\varepsilon\|x\|$.

If $\Phi$ is a positive linear map, then one might expect a Perron-Frobenius-type behaviour 
with regards to the spectrum of $\Phi$.
This has been known for several decades to be true. 

\begin{theorem}[Perron-Frobenius]\label{pf}\cite[Appendix 2.2]{Schaefer-tvs-book} 
If $\Phi:\A\rightarrow\A$ is a positive linear map, 
then the spectral radius of $\Phi$ is an element of the spectrum of $\Phi$.
\end{theorem}

\begin{definition}
The spectral radius of a positive linear map $\Phi:\A\rightarrow\A$ is called the \emph{Perron value} of $\Phi$.
\end{definition}

If $\Ep$ is a Schwarz channel on $\A$, then $\Ep$ is 
unital (Proposition \ref{basic positive facts}) and so $1\in  \mbox{\rm Sp}_{\rm p}\,\Ep$. 
But as $\|\Ep\|\leq 1$, the spectrum of $\Ep$ lies in the closed unit disc of $\mathbb C$, 
which implies that the spectral radius of $\Ep$ is $1$.
Thus, the \emph{spectral circle} for such $\Ep$ is the boundary $\mathbb T$ of the closed unit disc.

A theorem of Groh \cite{groh1981} states that if $\Phi$ is an irreducible unital
Schwarz map on a unital C$^*$-algebra, then 
the peripheral point spectrum $\mbox{\rm Sp}_{\rm p}\,\Phi\cap\mathbb T$ is a subgroup of $\mathbb T$.
In the case of Bures contractive channels, one has that this subgroup is trivial:

\begin{proposition} If $\Ep:\A\rightarrow\A$ is a Bures contractive Schwarz channel, then
\[
\mbox{\rm Sp}_{\rm p}\,\Ep\cap  \mathbb T=\{1\}.
\]
\end{proposition}

\begin{proof}
Suppose that $\omega \in \mathbb{T}$ is an eigenvalue of
$\mathcal{E}$ with (nonzero) eigenvector $x$. Because $\Ep$ preserves selfadjointness,
it is also true that
$\mathcal{E}(x^*)=\bar{\omega}x^*$. Therefore, 
\[
\mathcal{E}(x^*x)\geq \mathcal{E}(x)^*\mathcal{E}(x) =\bar{\omega} x^* {\omega}x =x^*x .
 \]
Applying the faithful trace $\tau$ to the inequality above yields the inequality
\[
0\leq\tau\left[\mathcal E(x^*x)-\mathcal E(x)^*\mathcal E(x)\right]=\tau(x^*x)-\tau(x^*x)=0.
\]
Thus, the positive element
$\mathcal E(x^*x)-\mathcal E(x )^*\mathcal E(x)$ has zero trace, which yields 
$\mathcal E(x^*x)=\mathcal E(x )^*\mathcal E(x)$. Hence, 
$x\in \mathcal S_\Ep =\mathcal M_\Ep$. 
However, Corollary \ref{trivial mult.dom} asserts that the multiplicative domain of a Schwarz channel is $\mathbb C\,1$; hence,
$x\in\ker(\Ep-I)\cap\ker(\Ep-\omega I)$, and so $\omega=1$. This proves that $ \mbox{\rm Sp}_{\rm p}\,\Ep\cap  \mathbb T=\{1\}$.
\end{proof}

It is well-known that if a completely positive channel $\Ep$ on $\M_d(\mathbb C)$ has a completely positive inverse, then 
$\Ep$ must be a unitary channel. As unitary channels are Bures isometries, intuition suggests that invertible Bures contractive channels
fail to be channels.

\begin{proposition} The linear inverse of a bijective Bures contractive channel is not positive.
\end{proposition}

\begin{proof} If $\Ep:\A\rightarrow\A$ is bijective channel, and if $\Ep^{-1}$ were also a channel, then $\Ep^{-1}$ 
satisfies the monotonicity property with respect to the fidelity function, implying for all distinct
$\sigma,\rho\in \mathcal{D}_{\tau}(\A)$ that
\[
F_{\tau}(\sigma,\rho)=F_{\tau}(\Ep^{-1}\circ\Ep(\sigma),\Ep^{-1}\circ\Ep(\rho))\geq F(\Ep(\sigma),\Ep(\rho))>F(\sigma,\rho),
\]
in contradiction to $F_{\tau}(\Ep(\sigma),\Ep(\rho))=F_{\tau}(\sigma,\rho)$. 
\end{proof}

\subsection{Nonscalar extreme points}

\begin{proposition} Assume that $\A$ is finite (i.e., $xy=1$ only if $yx=1$) and 
prime (i.e., $xay=0$, for every $a\in\A$, only if $x=0$ or $y=0$). 
If $\Ep$ is a Bures contractive Schwarz channel, and if $\A_1$ denotes the closed unit ball of $\A$,
then the image $\Ep(\A_1)$ does not contain 
any nonscalar extreme points of $\A_1$.
\end{proposition}
\begin{proof}
Without loss of generality, suppose there exists a nonscalar $a\in \A$ with $\|a\|\leq 1$  such that $\Ep(a)$ is an extreme point of the unit ball of $\A$. 
By Kadison's Theorem \cite{kadison1951}, $\Ep(a)$ is necessarily a partial isometry satisfying $(1-\Ep(a)^{*}\Ep(a))\A(1-\Ep(a)\Ep(a)^*)=\{0\}$. 
Using the quasi-transitivity and finiteness of $\A$, we get $\Ep(a)^{*}\Ep(a)=1=\Ep(a)\Ep(a)^*$. 
Since $\Ep$ is unital and
$aa^*\leq 1$, we have $\Ep(aa^*)\leq 1$ .
Now using the Schwarz inequality, 
\[
1=\Ep(a)\Ep(a)^*\leq \Ep(aa^*)\leq 1,
\]
resulting in $\Ep(aa^*)=\Ep(a)\Ep(a)^*$. Similarly, $\Ep(a^*a)=\Ep(a)^*\Ep(a)$. 
Hence, $a$ belongs to the multiplicative domain of $\Ep$. Since a Bures contractive map has $1$-dimensional multiplicative domain, 
this contradicts the assumption. Hence $\Ep(\A_1)$ can not contain any nonscalar extreme points of the closed unit ball of $\A$.  
\end{proof}

\section{Concluding Remarks}

The results herein confirm the intuition that one might have concerning (Bures) contractions of the density space -- the irreducibility of such
maps is one such example. With an eye toward applications, one may make several other statements, such as the following assertion concerning
error correction.

If $\Ep:\A\rightarrow\A$ is a channel, and if $\mathcal{C}$ is a nonempty susbet (of \emph{codes}) of the density space $\mathcal D_\tau(\A)$, then
$\Ep$ is \emph{correctable on $\mathcal C$} if there exists a channel $\mathcal R$ on $\A$ for which $\mathcal R\circ\Ep(\rho)=\rho$, for every $\rho\in\mathcal C$.
It was observed by Petz in \cite{petz} that the preservation of certain distinguishability measures between density operators is a sufficient condition for correctability on those operators;
it is, in fact, a necessary condition also, as shown by the result below.

 \begin{proposition} Bures contractive channels are not correctable on any set of codes of cardinality larger than $1$.
\end{proposition}

\begin{proof} Assume that $\mathcal{C}$ is a nonempty susbet of $\mathcal D_\tau(\A)$ consisting of at least two elements, and that there exists a channel 
$\mathcal R$ on $\A$ for which $\mathcal R\circ\Ep(\rho)=\rho$, for every $\rho\in\mathcal C$.

If $\sigma,\rho\in\mathcal C$ are distinct, then
\[
d_B^\tau(\sigma,\rho)>d_B^\tau\left( \Ep(\sigma),\Ep(\rho)\right) \geq d_B^\tau\left( \mathcal R\circ \Ep(\sigma),\mathcal R\circ\Ep(\rho)\right) =
d_B^\tau(\sigma,\rho),
\]
which is impossible. Thus, $\Ep$ is not correctable on any set $\mathcal{C}$ that has at least two elements.
\end{proof}

\section*{Acknowledgement}

The authors wish to thank the referee for comments regarding this paper and, in particular, for suggesting Corollary \ref{referee suggestion}.


 \end{document}